%% file: AAM_IEEE_Med_root.tex
\documentclass[letterpaper, 10 pt, conference]{ieeeconf}

\IEEEoverridecommandlockouts                              

\usepackage{cite}
\usepackage{amsmath,amssymb,amsfonts}
\usepackage{graphicx}
\usepackage{textcomp}
\usepackage{xcolor}
\def\BibTeX{{\rm B\kern-.05em{\sc i\kern-.025em b}\kern-.08em
    T\kern-.1667em\lower.7ex\hbox{E}\kern-.125emX}}

\usepackage[colorlinks=true, allcolors=blue]{hyperref}
\usepackage[english]{babel}
\usepackage{algpseudocode}
\usepackage{amsmath,amssymb,amsfonts}
\usepackage{amsthm}
\usepackage{cite}
\usepackage{gensymb}
\usepackage{hyperref}
    \hypersetup{
        colorlinks=true,
        linkcolor=blue,
        filecolor=magenta,      
        urlcolor=cyan,
        pdftitle={Overleaf Example},
        pdfpagemode=FullScreen,
        }
\usepackage{lipsum}
\usepackage{makeidx}  
\usepackage{mathptmx}
\usepackage{textcomp}
\usepackage{verbatimbox}

\input{Q_Shared_definitions}

\theoremstyle{definition}
\newtheorem{Algo}{Algoritm}

\title{\LARGE \bf
Exploring Noncommutative Polynomial Equation Methods for Discrete-Time Quaternionic Control*
}

\author{Michael Sebek$^{1}$
\thanks{*This work was co-funded by the European Union under the project ROBOPROX (reg. no. CZ$.02.01.01/00/22\_008/0004590$).}
\thanks{$^{1}$
Michael Sebek is with 
Department of Control Engineering,
Faculty of Electrical Engineering,
Czech Technical University in Prague,
Praha, Czech Republic
{\tt\small michael.sebek@fel.cvut.cz}}%
}
\newcommand{\IEEEAAMNotice}[2]{%
\noindent\begingroup\small
\hrule\vspace{4pt}
\noindent\textbf{Accepted manuscript (AAM).} Published in: #1. DOI: #2.\\[2pt]
\noindent\textcopyright{} 2025 IEEE. Personal use of this material is permitted. Permission from IEEE must be obtained for all other uses, in any current or future media, including reprinting/republishing this material for advertising or promotional purposes, creating new collective works, for resale or redistribution to servers or lists, or reuse of any copyrighted component of this work in other works.
\vspace{4pt}\hrule
\endgroup\vspace{10pt}
}

\begin{document}

\maketitle
\thispagestyle{empty}
\pagestyle{empty}

\IEEEAAMNotice{2025 IEEE 2025 33rd Mediterranean Conference on Control and Automation (MED), Tangier, Morocco,}{10.1109/MED64031.2025.11073531}
\begin{abstract}
We present new polynomial-based methods for SISO discrete-time quaternionic systems, highlighting how noncommutative multiplication modifies classical control approaches. Defining quaternionic polynomials via a backward-shift operator, we examine left and right fraction representations of transfer functions, showing that right zeros correspond to similarity classes of quaternionic matrix right eigenvalues. We then propose a feedback design procedure that generalizes pole placement to quaternions—a first approach using a genuine quaternionic polynomial equation. An illustrative example demonstrates its effectiveness. 
\end{abstract}

\section{INTRODUCTION}
Quaternions and quaternionic systems are gaining importance across science and engineering—especially in computer graphics, robotics, and control—because they effectively represent orientations and rotations \cite{Pereira06, MarkleyCrassidis14}. While many control systems papers treat quaternions as real-number vectors, here we view them as an extension of the complex numbers called hypercomplex numbers.

Owing to their noncommutative multiplication, quaternionic SISO systems differ from classical real or complex SISO systems. In some ways, they resemble MIMO systems, but they also exhibit unique behaviors that complicate design, rendering many classical tools inapplicable. This underscores the need for alternative methods of control analysis and design. Our approach offers the first method that utilizes a genuine quaternionic polynomial equation.

\section{BACKGROUND MATHEMATICS}

A quaternion has the form 
$a + b\i + c\j + d\k,$ where $ a, b, c, d $ are real numbers, and $ \i, \j, \k $ are imaginary units. 
Quaternions generalize the complex numbers but feature \emph{noncommutative multiplication}  (e.g., $\i\j = \k $ but $\j\i = -\k $). $\HH$ denotes their skew field. Quaternions $p,q$ are \emph{similar} if there exists a quaternion $u$ such that $ p = uqu^{-1},$ forming a \emph{similarity class} $[q].$  

Quaternionic matrices also differ from their complex counterparts. For instance, they lack a determinant, and \emph{left eigenvalues} $Av = \lambda v$ differ from \emph{right eigenvalues} $Av = v \lambda .$
For computational algorithms, see \cite{Wei18} and references therein.

A \emph{discrete-time quaternionic signal} is an infinite sequence 
$ x= \left\{x_n,x_{n+1}, x_{n+2}, \ldots  \right\},x_{k}\in\HH, n\in\ZZ .$ 
Using the backward-shift operator $d,$ defined as a sequence of zeros except for one at $n=1$, any signal can be expressed as a formal series 
$ x = x_n d^n + x_{n+1}d^{n+1}  + \ldots $
(see \cite{Kucera79}).
A sequence is \emph{causal} if it starts at $n=0$ and 
\emph{strictly causal,} if it starts at $n>0.$ A causal sequence is \emph{stable} if it converges to zero.

A \emph{quaternionic polynomial} is a causal sequence with finitely many nonzero terms.  
\begin{equation}
a_m d^m + a_{m+1} d^{m+1}+\ldots + a_nd^n, \;a_i\in\HH. 
\label{Eq:qpol}
\end{equation}
Here, the \emph{indeterminate} $d$ commutes with the coefficients $a_i$, and basic properties such as degree and order adhere to standard polynomial definitions.

However, due to noncommutativity, division, divisors, and greatest common divisors depend on the side of multiplication. Similarly, evaluation is noncommutative.
A quaternion $q$ is a \emph{right zero} of the polynomial $a(d)\in\HHd$ if 
$a_0+a_1q+\ldots+a_nq^n = 0$
and a \emph{left zero} if $a_0+qa_1+\ldots+q^n a_n= 0.$
A quaternionic polynomial of degree $n$ has either $n$ or infinitely many right zeros \cite{GordonMotzkin65,Lam91}.
Two polynomials  $a, b\in\HHd$  are right-zero coprime if they share no common right zero.

A polynomial $a\in\HHd$ is \emph{stable} if its inverse $a^{-1}$ is a stable quaternionic sequence. 
This requires that all right zeros of $a$ have norms greater than 1, a direct consequence of using the backward-shift operator $d$ instead of the forward-shift 
$z.$

Fractions of real polynomials are a standard tool in control theory. However, quaternionic polynomial fractions behave more like fractions of real polynomial matrices due to the noncommutativity of multiplication. This requires distinguishing between left and right fractions: A \emph{left quaternionic polynomial fraction} is
\begin{equation}
    a_{\ell}^{-1} b_{\ell} , \quad  a_{\ell},b_{\ell} \in\HHd,  \label{Eq:lqpf}
\end{equation}
where  $a_{\ell}$ is the \emph{left denominator} and $b_{\ell} $ is the \emph{left numerator.} A \emph{right quaternionic polynomial fraction} is
\begin{equation}
    a_{\err} b_{\err}^{-1}  , \quad  a_{\err},b_{\err} \in\HHd , \label{Eq:rqpf}
\end{equation}
where $a_{\err}$ is the \emph{right numerator} and $b_{\err} $
 is the \emph{right denominator.}
Typically, the polynomials in \eqref{Eq:rqpf} must be right coprime, and those in \eqref{Eq:lqpf} left coprime.

We use the backward-shift \emph{indeterminate} (operator) $d$ instead of the quaternionic \emph{variable} $\zi$ to avoid noncommutativity issues and as the quaternionic $z$-transform is not yet fully developed.
%
\section{QUATERNIONIC SYSTEMS}
State-space models for discrete-time single-input, single-output systems over quaternions are commonly expressed as
\begin{equation}
\begin{aligned} 
    x(k+1) &= F x(k) + G u(k) , \\
y(k) &= H x(k) +J u(k),  
\end{aligned}    
\label{Eq:ssmodel}
\end{equation}
where, all signals  ($u(k),y(k)\in\HH, x(k)\in\HHn,$) and matrices ($F\in\HHnn, G\in\HHnm, H\in\HH^{l\times n},J\in\HH^{l\times m}$) have quaternion entries.

The causal scalar quaternionic sequence 
\begin{equation}
 S = S_0 + S_1d + S_2 d + \cdots  ,     \label{Eq:tf1a}   
\end{equation}
with
\begin{equation}
\begin{aligned}
    S_0 &= J , \\
    S_k &= H F^{k-1}G, \;k = 1,2,\ldots  ,
\end{aligned}   \label{Eq:tf1b}   
\end{equation}
serves as the system's transfer function. It can also be expressed compactly as
\begin{equation}
S = J + d H (I_n-d F)^{-1}G.  \label{Eq:tf2}
\end{equation}
Any quaternionic transfer function can be represented as a quaternionic polynomial fraction
\begin{equation}
S = p_{\ell}^{-1} q_{\ell} = q_{\err} p_{\err}^{-1} , \label{Eq:tffrac}
\end{equation}
where $p_{\ell}, q_{\ell} $ are left coprime and $p_{\err}, q_{\err} $ are right coprime. These fractions are unique up to multiplication by a nonzero quaternion from within.

The inverse matrix $(I_n-d F)^{-1}$ in \eqref{Eq:tf2} 
is not a quaternionic polynomial, making its computation intricate. However, we can avoid this by expressing 
$S$ in \eqref{Eq:tf2} as a polynomial fraction from. Below, we outline the procedure for obtaining the left fraction.

\begin{Algo} \label{Alg:1}
\mbox{}\\
\textbf{Given:} $F,G,H,J.$\\
\textbf{1)} Find $p_\ell\in\HHd, Q_\ell \in\HHggd{1}{n}$ such that
  $$
  p_\ell^{-1}Q_\ell = d H \left( I_n - d F \right)^{-1}.
  $$
\textbf{2)} Compute the numerator
  $$
  q_\ell = p_\ell J + Q_\ell G \in \HHd .
  $$
\textbf{Result:} $p_\ell, q_\ell \in\HHd$ satisfying
  $
  S = p_{\ell}^{-1}(d) q_\ell(d).
  $
\end{Algo}
By ensuring coprimeness at each step, the resulting 
 $ p_{\ell}, q_\ell$ remains left coprime.
 Similarly, the right polynomial fraction representation of $S$ can be directly obtained. 
 
 Finally, the right fraction $ q_{\err} p_{\err}^{-1}$ follows from the left fraction $p_{\ell}^{-1} q_{\ell} $
by computing the right null space of the matrix
$$\begin{bmatrix}
    p_\ell & - q_\ell 
\end{bmatrix}.$$

A realization $(E, F,G,H)$ is minimal if it has the fewest states among all that realize 
$S.$ Only such a minimal realization uniquely corresponds to $S;$ nonminimal ones harbor hidden (unreachable or unobservable) modes. If the realization is minimal and fractions coprime, we have 
$$
n = \max(\deg p_\ell,\deg q_\ell) = \max(\deg p_\err, \deg q_\err) ,
$$
since we operate in the backward-shift  $d.$

In classical systems, if the realization is minimal and the fractions in \eqref{Eq:tffrac} are coprime, then all three polynomials $\det (I_n - d F),\;  p_\ell$ and $p_\err$
share the same zeros, which also coincide with the reciprocals of the eigenvalues of 
F. However, in quaternionic systems, the usual notions of determinant and characteristic polynomials do not directly apply. Nevertheless, as we will show, although the right zeros of 
$p_\ell$ and $p_\err$ may differ, they remain similar—and they are also similar to the reciprocals of the right eigenvalues of $F.$

\begin{The}\label{The:zeroslrf}
Let $p_\ell,q_\ell\in\HHd$ be left coprime and $p_\err,q_\err\in\HHd$ be right coprime polynomials such that 
$$p_{\ell}^{-1} q_{\ell} = q_{\err} p_{\err}^{-1}.$$ 
Then, for every right zero $r_{p_{\ell}}$
of $p_{\ell},$ there exists a right zero $r_{p_{\err}}$ of $p_{\err}$ such that
$$
r_{p_{\err}} \in [ r_{p_{\ell}} ].
$$
Conversely, for every right zero $r_{p_{\err}}$
of $p_{\err}$ there exists a right zero $r_{p_{\err}}$ of $p_{\err}$ such that
$$
r_{p_{\ell}} \in [ r_{p_{\err}} ].
$$
\end{The}
\begin{proof}
Since the polynomial fractions are equal, we have
\begin{equation}
    p_\ell q_\err = q_\ell p_\err . \label{Eq:l2r3}
\end{equation}
The zeros of $p_\ell$ are similar (though not necessarily equal) to the zeros of the left-hand side of this equation and, thus, similar to those of the right-hand side. Given that $p_\ell,q_\ell\in\HHd$ are left coprime, the zeros of $p_\ell$  must be similar to the zeros of $p_{\err}.$ The reverse direction follows by an analogous argument.
\end{proof}
The reverse process—finding a state realization for a given transfer function—is more complex for quaternionic systems than for classical ones. Despite their structural similarity, quaternionic realizations require additional steps.

\begin{The}[Quaternionic state realization]
\label{The:ssreal}
Let $S$ be a quaternionic transfer function expressed as a left quaternionic polynomial fraction 
 \begin{equation*}
  S = a_\ell^{-1}b_\ell =
   \left( 1 + a_1 d + \cdots + a_n d^n \right)^{-1}
    \left( b_0 + \cdots b_n d^n \right) ,
       \label{Eq:lf}
 \end{equation*}  
where $a_0 = a_\ell (0) = 1.$ 
Then, a corresponding state-space realization 
$(E,F,G,J)$ satisfying \eqref{Eq:tf2} is given by
\begin{equation*}
\begin{aligned}
F &=
\begin{bmatrix}
    0 & 1& 0 & \ldots & 0\\
    0 & 0& 1 & \ldots & 0\\
    \vdots & \vdots & \ddots & \ddots & \vdots\\
    0 & 0& 0 & \ldots & 1\\
    -\hat{a}_n & -\hat{a}_{n-1}& \hat{a}_{n-2}&\ldots &-\hat{a}_1
\end{bmatrix} ,
&G =
\begin{bmatrix}
    0\\ 0\\ \vdots \\ 0\\ 1
\end{bmatrix} ,
\\[2mm]
H &=
\begin{bmatrix}
\hat{b}_n,\hat{b}_{n-1} , \cdots ,\hat{b}_1 
\end{bmatrix} ,
&J =
\begin{bmatrix}
    b_0
\end{bmatrix}.
\end{aligned}
\label{Eq:realization}
\end{equation*}
Here, $\hat{a}_i$ are the coefficients of the polynomial 
$\hat{a}\in\HHd$, obtained from the left-to-right fraction conversion
\begin{equation}
   \hat{b}(d)\hat{a}(d)^{-1}= a(d)^{-1}\tilde{b}(d) ,
    \label{Eq:rf}
\end{equation}
where $\hat{a}_0 = 1$ and $\tilde{b}(d)\in\HHd$ is given by
\begin{equation}
     \tilde{b}(d) = \tilde{b}_1d + \cdots + \tilde{b}_nd^{-n} 
     = b_\ell(d)-b_0 a\ell(d).
\end{equation}
\end{The}
\begin{proof}
The proof follows the classical system approach but requires additional polynomial fraction conversions and technical details. Due to space constraints, we omit these here.
\end{proof}
This realization corresponds to the quaternionic counterpart of the classical controllable canonical form. While the equations may resemble those in classical systems, they integrate coefficients from multiple polynomial fractions.
\begin{The} \label{The:zeroseigen}
If $r_i$ is a right zero of the denominator polynomial $a_\ell,$ then
$$
r_i^{-1}\in [\lambda_{\err,i}] , 
$$ 
where $\lambda_{\err,i}$ is some right eigenvalue of $F,$ for $i=1,\ldots, \deg a_\ell .$ Moreover, if
 $\deg b_\ell - \deg a_\ell = k > 0$, then $F$ has additional $k$ additional zero left (and hence also right) eigenvalues.
\end{The}
\begin{proof}
First, note that all $r_i\neq 0$ as $a_0 = a_\ell (0) = 1.$
 By Theorem \ref{The:zeroslrf}, each $r_i$ is similar to a right zero $\hat{r}_i$ of $a_\err.$

Since $F$ is in a bottom companion form, if  $a_\ell$ were a polynomial in the forward shift 
$z,$ then its right zeros would correspond to left eigenvalues of $F$ (see \cite{SerodioPereiraVitoria01}). However, since $a_\ell$ is expressed in the backward shift  $d=\zi,$ each of its right zeros is the reciprocal of a left eigenvalue of $F.$
By \cite{SerodioPereiraVitoria01}, each left eigenvalue of the companion matrix represents an entire class of right eigenvalues, confirming the first claim.

For the second claim, note that since the polynomial fractions are coprime, their denominators share the same degree. If $\hat{a}_n=\hat{a}_{n-1}=\cdots=\hat{a}_{n-k}=0,$ the $F$ has $k$ left eigenvalues equal to zero, which implies it also has $k$ right eigenvalues equal to zero.
\end{proof}

\section{STABILITY}
Stability is a fundamental property of linear systems. A system $(F,G,H,J)$ defined by \eqref{Eq:ssmodel} is (internally, asymptotically) stable if, for any initial state $x_0\in\HHn,$ the initial state response
$$S_{x_0} = \{ x_0, F x_0, F^2 x_0,\ldots \}$$ 
converges to zero. 
Similar to \eqref{Eq:tf2}, this response can be expressed as a left polynomial matrix fraction
\begin{equation}
 S_{x_0} = (I_n-d F)^{-1}x_0 .  
        \label{Eg:isr}
\end{equation}
Clearly, $S_{x_0}$ converges to zero for any $x_0$ if and only if the sequence
 $(I_n-d F)^{-1} =  \{ I_n, F, F^2,\ldots \}$
converges to zero. It is well known that this occurs if and only if all right eigenvalues of the quaternionic matrix $F$ satisfy
$$
| \lambda_{\err, i}| > 0.
$$ 

If $(I_n-d F)^{-1}$ converges to zero, then the transfer function $S=H(I_n-d F)^{-1}G$ converges to zero. 
The converse holds if and only if the system has no unstable hidden modes.

In the polynomial fraction representation
\begin{equation}
S =  p_{\ell}^{-1} q_{\ell} = q_{\err} p_{\err}^{-1} = J + d H (I_n-d F)^{-1}G,
\label{Eq:alltf}
\end{equation}
 the right zeros of the denominator polynomials $p_{\ell}$ and $ p_{\err} $ 
are similar but not necessarily equal to the right eigenvalues of $F.$ Fortunately, this difference does not affect stability analysis.
\begin{The}
Let the realization $(F,G,H,D)$ be minimal, and assume the polynomial fractions in \eqref{Eq:alltf} are coprime. Then, all right eigenvalues of 
$F$ are stable if and only if the denominators 
$ p_{\ell}$ and $ p_{\err}$ are stable polynomials.
\end{The}
\begin{proof}
The proof follows from Proposition 1, noting that all quaternions in a similarity class share the same norm.
\end{proof}
\begin{Pro}
$$
|\lambda | > 1, q\in\HH \quad\Rightarrow\quad  |\mu | > 1,\; \forall \mu \in [\lambda],\mu\in\HH .
$$
\end{Pro}
\begin{proof}
This is straightforward, as all quaternions within a similarity class share the same norm (see also \cite{PereiraVettori06}).
\end{proof}
%
\section{FEEDBACK SYSTEM}
Consider the output feedback system
\begin{equation}
     y = Su + w  \eqand   u = -Ry + v ,
\end{equation}
where $S$ is the plant's transfer function, $R$ is the controller's transfer function, $u, y$ are the plant input and output sequences. 
The external inputs  $v$ and $w$ represent responses to initial conditions, external disturbances, etc. As in classical systems (see \cite{Kucera79}), the closed-loop transfer functions from $v$ and $w$ to  $y$ is given by
\begin{align}         \label{Eq:cltfy}
        y &= (1+SR)^{-1}Sv + (1+SR)^{-1}w  .
\end{align}
Expressing the plant and controller transfer functions as polynomial fractions
\begin{equation}
  S=a_\ell^{-1}b_\ell, \quad R = q_\err p_\err^{-1} ,     
\end{equation}
where $a_\ell, b_\ell,q_\err p_\err \in\HHd,$ 
relation \eqref{Eq:cltfy} transforms into:
\begin{equation}
\begin{aligned} 
        y =& p_\err (a_\ell p_\err +b_\ell q_\err )^{-1}b_\ell v 
               \\
               &+ p_\err (a_\ell p_\err +b_\ell q_\err )^{-1} a_\ell  w .
\end{aligned}    
\label{Eq:cltfy2}
\end{equation}
Defining the feedback denominator polynomial
\begin{equation}
  c = a_\ell p_\err +b_\ell q_\err  \in \HHd ,\label{Eq:cdenpol}
\end{equation}
we rewrite \eqref{Eq:cltfy2} as
\begin{equation}   
    y = p_\err c^{-1}b_\ell  v + p_\err c^{-1}a_\ell  w .
\end{equation}
Converting to left fractions 
$$g^{-1}h = p_\err c^{-1} ,$$ 
and denoting $h_v=hv,h_w=hw,$
the closed-loop transfer functions take the form
\begin{equation}
            y =  g^{-1}h_v v + g^{-1}h_w w .
\end{equation}
Given $a_\ell, b_\ell$ (the plant) and fixing closed-loop denominator polynomial $c$ turns \eqref{Eq:cdenpol} into a linear quaternionic polynomial equation with unknowns  $p_\err$ and $q_\err$ (the controller).
\section{POLYNOMIAL EQUATIONS}
Scalar linear polynomial equations are fundamental in classical systems, yet quaternionic polynomials require distinguishing among right-, left-, and two-sided forms due to noncommutative multiplication. For simplicity, we focus on a single form. Specifically, we consider the \emph{quaternionic polynomial equation}
\begin{equation}
a(d)x(d)+b(d)y(d)=c(d),                          \label{Eq:spe}
\end{equation}
where $a(d),b(d)$ and $c(d)$ are given quaternionic polynomials and $x(d),y(d)$ are unknowns. 
A solution is any pair $(x(d),y(d))$ from $\HHd$ that satisfies the equation.

\begin{The}[Solvability]

The equation (\ref{Eq:spe}) has a solution if and
only if the greatest common left divisor of $a$ and $b$ is also a left divisor of $c.$ 
\end{The}
\begin{proof}
Let $g=\gcld (a,b)$ so that $a=g\tilde{a}$ and $b=g\tilde{b}$ for some left coprime polynomials $\tilde{a},\tilde{b}\in\HHd.$
Suppose $x,y\in\HHd$ is a solution of (\ref{Eq:spe}). 
Substituting $a=g\tilde{a}$ and $b=g\tilde{b}$ into \eqref{Eq:spe} gives
$$
g(\tilde{a}x+\tilde{b}y)=c .
$$ 
This implies that $g$ is a left divisor of $c.$ 

Conversely, if $g$ is a left divisor of $c$ then $c=g\tilde{c}$ for some $\tilde{c}\in\HHd.$
By a quaternionic version of Bézout's identity \cite{DamianoGentiliStruppa}, there exist polynomials $p,q\in\HHd$ such that
$$
ap+bq=g.
$$ 
Multiplying it on the right by $\tilde{c},$ yields
$$
a(p\tilde{c})+b(q\tilde{c})=g\tilde{c}=c.
$$ 
Thus, $x=p\tilde{c}$ and $y=q\tilde{c}$  form a solution to \eqref{Eq:spe}.
\end{proof}
\begin{The}[General solution]
Let $x',y'\in\HHd$ be a particular solution of equation \eqref{Eq:spe}. Then, the general solution is
\begin{equation}
\begin{array}{l}
x = x' - b_\err t ,\\
y  = y' + a_\err t   ,               
\end{array}                                              \label{Eq:gsspe}
\end{equation}
where $t\in\HHd$ is an arbitrary quaternionic polynomial parameter, and
$a_\err, b_\err \in\HHd$ are right-coprime quaternionic polynomials satisfying
\begin{equation}
 a^{-1}b = b_\err a_{\err}^{-1}.  
 \label{Eq:lrconv}
\end{equation}
\end{The}
\begin{proof}
Since $ ax'+by'=c,$ subtracting this from \eqref{Eq:spe} gives the homogeneous equation
\begin{equation}
a(x-x') + b(y-y')=0.    \label{Eq:homo}
\end{equation}
Thus, we seek all polynomial matrices 
$M\in\HHggd{2}{1}$ satisfying
$$
\begin{bmatrix}
    a & b
\end{bmatrix}
M=0.
$$
Since 
$
\begin{bmatrix}
    a & b
\end{bmatrix}
$
is $1\times 2$ of rank one, $M$ must also have rank one and can be written as
$$
M=Z t
$$
where $Z$ is a full rank $2\times 1$ matrix and $t\in\HHd$ is arbitrary.
Partitioning $Z$ compatibly with 
$
\begin{bmatrix}
    a & b
\end{bmatrix}, 
$
we get
$$
Z=
\begin{bmatrix}
    - b_\err \\ a_\err
\end{bmatrix},
$$
where $a_\err,b_\err$ are right coprime polynomials from \eqref{Eq:lrconv}. Substituting into \eqref{Eq:homo}, we get
$$
x - x' = -b_r t, \quad y - y' = a_r t, \quad 
$$
which completes the proof.
\end{proof}
\begin{The}[Minimal solutions]
If equation \eqref{Eq:spe} is solvable, then among all solutions, there exists one where $x$ has the minimum degree, satisfying
    $$
    \deg x < \deg b_\err.
    $$
Similarly, there exists (possibly another) solution where $y$ has the minimum degree, satisfying
    $$
    \deg y < \deg a_\err.
    $$
\end{The}
\begin{proof}
Let $x',y'$ be any particular solution.
Performing right polynomial division on $x'$ by $\tilde b_\err,$ we obtain 
\begin{equation}
   x'= b_\err q + r    ,             \label{Eq:quotientreminder}
\end{equation}
where $q$ is the quotient and $r$ is the reminder, with $\deg r  <  \deg b_\err.$
Substituting $q$ as the parameter $t$ in the general solution
$$
x = x' - b_\err t = x' - b_\err q .
$$
Since (\ref{Eq:quotientreminder}) implies  $x' - b_\err q = r,$ we get x = $r,$ which has the minimal possible degree.
The construction for minimizing 
$\deg y$ follows the same reasoning.
\end{proof}

\section{DESIGN PROCEDURE}
We use quaternionic polynomial equations to place the zeros of a closed-loop transfer function denominator in specified similarity classes. Specifically, we design an output feedback controller so that the denominator polynomial has right zeros similar to desired values—mirroring classical pole placement in a quaternionic setting.
\begin{Algo}[Design procedure]
\mbox{}\\
\textbf{1) Construct the desired denominator:}
Given the desired positions  $r_i\in\HH$, form the polynomial
\begin{equation}
    c(d)= (d-r_1)(d-r_2)\ldots  \; .
\end{equation}
\textbf{2) Compute the plant's transfer function:}  
If a state-space model \eqref{Eq:ssmodel} is given, use Algorithm 1 to compute the left polynomial fraction
\begin{equation}
    S = a_{\ell}^{-1} b_{\ell} =  J + d H (I_n-d F)^{-1}G.  \label{Eq:tf4}
\end{equation}
\textbf{3) Solve the polynomial equation:}  
Form the polynomial equation:
\begin{equation}
    a_\ell p_\err +b_\ell q_\err = c 
    \label{Eq:spe2}
\end{equation}  
and solve it for unknown polynomials $p_\err $ and $ q_\err .$\\
\textbf{4) Construct the controller transfer function:}  
\begin{equation}
    R =  q_\err p_\err^{-1}
\end{equation}
\textbf{5) Obtain the controller's state-space model:}  
If needed, compute its realization using Theorem \ref{The:ssreal}.\\

With this controller, the closed-loop transfer functions become
\begin{equation}
    T_w = p_\err c^{-1}a_\ell =  g^{-1}h_w, \quad
    T_v = p_\err c^{-1}b_\ell = g^{-1}h_v .
\end{equation}
Here, $c$ and $g$ have right zeros similar to the prescribed $ r_i$'s.
\end{Algo}
Unlike in classical systems, zeros cannot be placed exactly at the desired positions in $\HH$, only within their similarity classes. However, this is sufficient, as all elements in a similarity class share key properties. They are all either stable or unstable.

\begin{Ex}
Consider a plant defined by the state-space model  \eqref{Eq:ssmodel} with
\begin{equation*}
    F = 
    \begin{bmatrix}
    1 & \i \\ \j  & \k 
    \end{bmatrix},\;
    G = 
    \begin{bmatrix}
    \i \\ 0 
    \end{bmatrix},\;
    H = 
    \begin{bmatrix}
    1 & 0 
    \end{bmatrix},\;
    J = 0 .
\end{equation*}
The standard right eigenvalues of  $F$ are 
 $\lambda_1=1.4+0.37\i, \lambda_2= -0.37+1.4\i$ both having 
 $ |\lambda_1|= |\lambda_2| =  1.4142 > 1$  indicating instability.

Following the outlined design steps:\\
\textbf{1)} Select stable values $r_1 = 3, r_2 = 4$ in $d=\zi$,  forming the polynomial
$$
    c(d)= (d-3) (d-4) = 12 - 7d + d^2 .
$$
The choice of real zeros is possible but not necessary. It will have interesting consequences. \\
\textbf{2)} Using Algorithm 1, the left polynomial fraction representation of \eqref{Eq:tf4}  gives
\begin{align*}
a_\ell(d) &= 1 - (1-0.31\i-0.89\j-0.35\k)d \\
& \hspace{20mm}- (0.61\i+1.8\j+0.7\k)d^2 ,\\
b_\ell(d) &=  d + (0.31i+0.89j+0.35k)d^2 .
\end{align*}
The right zeros of $a_\ell$ are computed as
 $ z_1^{-1} =-0.37-0.42\i-1.2\j-0.48\k, \; 
 z_2^{-1} = 1.4+0.11\i+0.32\j+0.13\k$ which are similar to the reciprocals of  $\lambda_1, \lambda_2$, confirming correctness.\\
\textbf{3)} Solving \eqref{Eq:spe2} yields
\begin{align*}
p_\err &= 12 - (11-5.8\i-17\j-6.6\k)d ,\\
q_\err(d) &= 30\i-14\j+10\k - (37\i -19\j+15\k )d   .
\end{align*}
This defines the controller transfer function
$R = q_\err p_\err^{-1}.$

Expressing the closed-loop transfer function from $w$ to $y$ as a left fraction  $g^{-1}h_w$ we obtain
\begin{align*}
g(d) &= 0.095-0.12\i-0.36\j-0.14\k \\
    &\hspace{8mm} - (0.056-0.073\i-0.21\j-0.083\k)d \\
    &\hspace{8mm}  + (0.008-0.01\i-0.03\j-0.012\k)d^2 , \\
h_w(d) &= 0.095-0.12\i-0.36\j-0.14\k \\
   & \hspace{8mm} + (0.87+0.31\i+0.91\j+0.36\k)d \\
   &\hspace{8mm}  - (1.9+0.048\i+0.14\j+0.055\k)d^2 \\
   &\hspace{4mm}  + (1.1-0.34\i-0.99\j-0.39\k)d^3   .
\end{align*}

Despite being a quaternionic polynomial, 
$g$ has real zeros 3 and 4, achieving exactly the desired placement.

The resulting closed-loop state-space model is
\begin{align*}
    F_c &= 
    \begin{bmatrix}
      0  &   1    &     0  \\ 
      0  &   0   &      1   \\
     0  &  -0.083   &  0.58
    \end{bmatrix}   ,
   & G_c &=
    \begin{bmatrix}
    0 \\ 0 \\ 1
    \end{bmatrix}   , \\[1mm]
     H_c &= 
    \begin{bmatrix}
    3.2+0.56\i+1.6\j+0.64\k  \\  -0.75-1.4\i-4\j-1.6\k    \\ -1.3+0.79\i+2.3\j+0.9\k 
    \end{bmatrix}^T ,
    &J_c &= 1   .
\end{align*}
Since $c(d)$ was chosen as real in this example,
$F$ with eigenvalues $\lambda_1(F_c) =  0.33,
  \lambda_2(F_c) =    0.25, \lambda_3(F_c) =  0
$. As expected $\lambda_1(F_c) = 1/r_1,$ $\lambda_2(F_c) = 1/r_2,$ The third eigenvalue is zero since 
$c(d)$ was selected with a degree one less than the closed-loop system order.

This example shows how quaternionic polynomial equations enable pole placement in quaternion feedback systems; Figure 1 illustrates the closed-loop response.
\begin{figure}[thpb]
      \centering
      \includegraphics[scale=0.15]{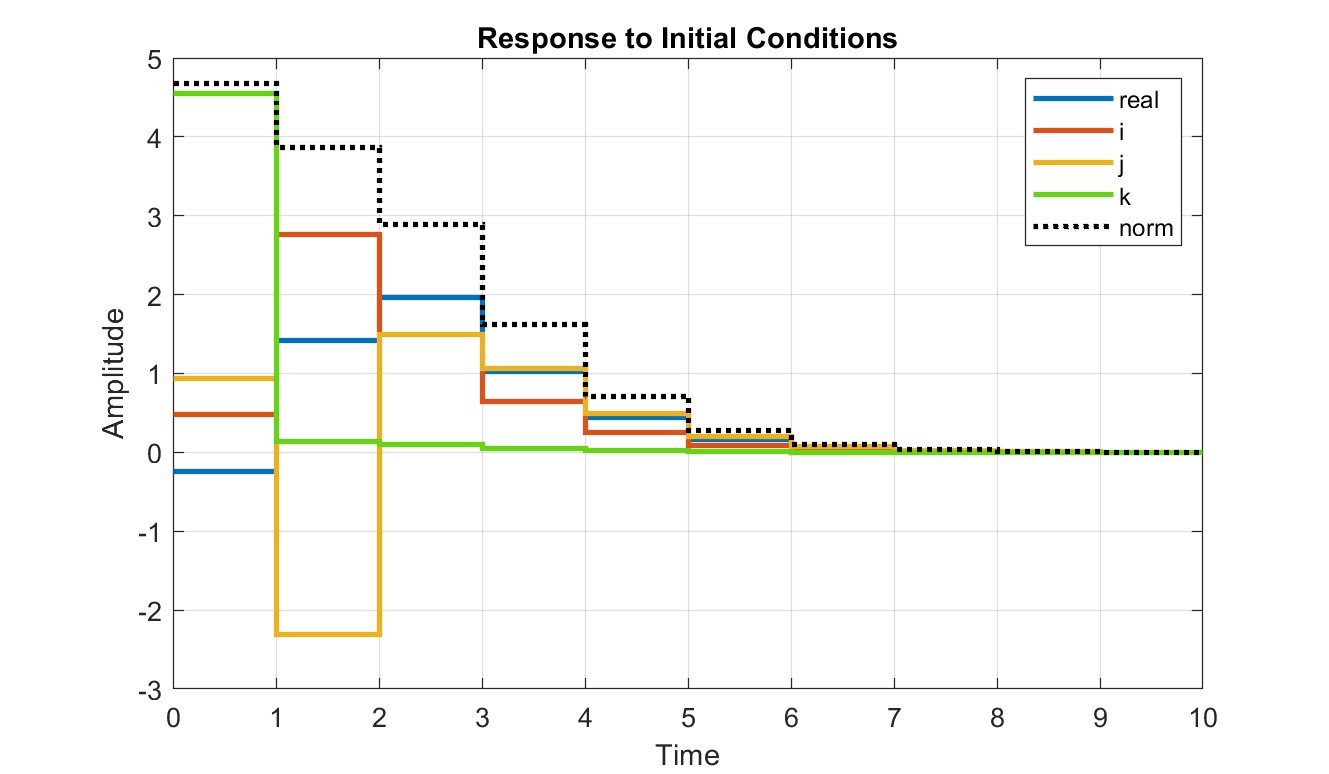}
      \caption{Feedback system response to a random initial state.}
      \label{figurelabel}
\end{figure}
\end{Ex}
\addtolength{\textheight}{-12cm}   


\end{document}

%% file: Q_Shared_definitions.tex
\theoremstyle{definition}
\newtheorem{Def}{Definition}
\theoremstyle{definition}
\newtheorem{The}{Theorem}
\newtheorem{Pro}{Proposition}
\theoremstyle{definition}
\newtheorem{Lem}{Lemma}
\theoremstyle{definition}

\theoremstyle{remark}

\theoremstyle{remark}
\newtheorem{Ex}{Example}
\theoremstyle{remark}

\usepackage{mdframed}

\DeclareFontEncoding{LS1}{}{}
\DeclareFontSubstitution{LS1}{stix2}{m}{n}
\DeclareMathAlphabet{\mathscr}{LS1}{stix2scr}{m}{n}
\newcommand{\err}{\ensuremath{ \mathscr{r} }}
\renewcommand{\i}{\mathbf{i}}
\renewcommand{\j}{\mathbf{j}}
\renewcommand{\k}{\mathbf{k}}

\newcommand{\ZZ}{\mathbb{Z}}

\newcommand{\HH}{\mathbb{H}}

\newcommand{\HHd}{\mathbb{H}[d]}
\newcommand{\HHn}{\mathbb{H}^{n}}
\newcommand{\HHnn}{\mathbb{H}^{n\times n}}

\newcommand{\HHnm}{\mathbb{H}^{n\times m}}

\newcommand{\HHggd}[2]{\mathbb{H}^{#1\times #2}[d]}
\newcommand*{\eqand}{\ensuremath{\quad\textrm{ and }\quad}}

\DeclareMathOperator{\gcld}{gcld}

\newcommand{\zi}{z^{-1}}
